\documentclass{article}
\usepackage[utf8]{inputenc}
\usepackage{amsmath,amsfonts}
\usepackage{gensymb}
\usepackage[left=1in, right=1in, top=1in, bottom=1in]{geometry}
\usepackage{graphicx}
\usepackage{tikz}
\usetikzlibrary{positioning}
\usepackage{mathtools}
\usepackage{amsthm}
\usepackage[numbers]{natbib}
\usetikzlibrary{arrows,automata,calc}
\usepackage{nccmath}
\usepackage{xspace}
\usepackage{xcolor}
\definecolor{MyBlue}{rgb}{0.12, 0.12, 0.76}
\usepackage[colorlinks,allcolors=MyBlue]{hyperref}
 
\usepackage{thmtools}
\usepackage{thm-restate}

\newtheorem{theorem}{Theorem}[section]
\newtheorem{lemma}{Lemma}[section]

\newtheorem{definition}{Definition}[section]

\newtheoremstyle{named}{}{}{\itshape}{}{\bfseries}{.}{.5em}{\thmnote{#3's }#1}
\theoremstyle{named}

\usepackage{algorithm}
\usepackage{algpseudocode}
\usepackage{algorithmicx}
\let\oldReturn\Return
\renewcommand{\Return}{\State\oldReturn}
\algtext*{EndWhile}
\algtext*{EndIf}
\algtext*{EndForAll}
\algtext*{EndFor}
\algtext*{EndFunction}

\DeclareMathOperator*{\argmax}{arg\,max}

\allowdisplaybreaks 

\newcommand\X{\mathcal{X}}
\newcommand\suc{\texttt{succ}}
\newcommand\pred{\texttt{pred}}
\newcommand\V{\mathbf{v}}

\title{Almost Envy-free Repeated Matching in Two-sided Markets}

\author{Sreenivas Gollapudi\\ sgollapu@google.com \and Kostas Kollias\\ kostaskollias@google.com \and Benjamin Plaut\\ bplaut@cs.stanford.edu}

\date{}

\begin{document}

\maketitle

\begin{abstract}
A two-sided market consists of two sets of agents, each of whom have preferences over the other (Airbnb, Upwork, Lyft, Uber, etc.). We propose and analyze a repeated matching problem, where some set of matches occur on each time step, and our goal is to ensure fairness with respect to the cumulative allocations over an infinite time horizon. Our main result is a polynomial-time algorithm for additive, symmetric ($v_i(j) = v_j(i)$), and binary ($v_i(j) \in \{a,1\}$) valuations that both (1) guarantees \emph{envy-freeness up to a single match} (EF1) and (2) selects a maximum weight matching on each time step. Thus for this class of valuations, fairness can be achieved without sacrificing economic efficiency. This result holds even for \emph{dynamic valuations}, i.e., valuations that change over time. Although symmetry is a strong assumption, we show that this result cannot be extended to asymmetric binary valuations: (1) and (2) together are impossible even when valuations do not change over time, and for dynamic valuations, even (1) alone is impossible. To our knowledge, this is the first analysis of envy-freeness in a repeated matching setting.
\end{abstract}


\section{Introduction}\label{sec:intro}

Recent years have seen a dramatic increase in electronic marketplaces, both in quantity and scale. Many of these are \emph{two-sided} markets, meaning that the market makes matches between two sets of agents (homeowners and guests for Airbnb, employers and workers for Upwork, drivers and riders for Lyft and Uber, etc.), each of whom has preferences over the other. This is in contrast to traditional resource allocation (cake cutting, Fisher markets, auctions, etc.) where only one side of the market has preferences. Although envy-freeness and relaxations thereof have been studied extensively in one-sided resource allocation (this research area is typically referred to as ``fair division"), we are aware of just one paper considering envy-freeness for two-sided preferences~\cite{Patro2020}.\footnote{Although~\cite{Patro2020} is conceptually similar, it is in different setting of \emph{recommendation algorithms}, and so it is technically quite different.}

There are two primary motivations for our work. The first is to simply study fair division for two-sided preferences. The second is that in some ways, two-sided electronic marketplaces like Airbnb, Upwork, Lyft, and Uber are actually in a better position to impose fairness than one-sided marketplaces. The reason is that most one-sided markets are \emph{decentralized}, in the sense that a seller offers different goods at different prices, buyers peruse the wares at their leisure, and make individual decisions about what they wish to purchase. On the contrary, most two-sided markets operate by way of matches mediated by a centralized platform, giving the platform the ability to affect the outcomes of the system. Indeed, on Lyft and Uber, an automated central authority has almost complete control over the matches, giving the algorithm tremendous power over the outcomes for each individual agent. The power dynamic between the platform and the participating agents makes it even more important to ensure that the matching algorithms are fair to each agent.

\subsection{Repeated matching}

A crucial element of marketplaces is \emph{repeated} matching. Agents do not receive all of their matches at once; typically, an agent can only process a few matches at a time (a driver can only fit so many riders in the car, a worker can only handle so many contracts at once). Only after an agent completes some of her current matches can she be given new matches. This motivates a model where on each time step, an irrevocable matching decision must be made, and we expect fairness with respect to the cumulative matching at each time step. We consider an infinite time horizon and a finite set of agents, so we must allow the same pair of agents to be matched multiple times; thus each agent's cumulative set of matches will be a multiset.

A vital aspect of any repeated setting is that preferences can change. In some cases, preferences may change in direct response to the matches an agent receives: if a driver is matched with a rider who wishes to go to location X, once that ride is completed, the driver will prefer riders whose pickups are close to X. In other cases, agents may desire variety among matches: Airbnb guests may not wish to vacation in the same area every time. Additionally, preferences may simply drift over time. We refer to valuations that change over time as \emph{dynamic} valuations.

\subsection{Fairness notions}

For one-sided fair division with indivisible items, full envy-freeness is impossible: for two agents and a single item, one must receive the item and the other agent will be envious. The same issue applies in our setting: if every agent on one side of the market is interested in the same agent on the other side, no algorithm can guarantee envy-freeness.

One solution is to consider relaxations of envy-freeness, such as \emph{envy-freeness up to one good} (EF1). An outcome is EF1 if whenever agent $i$ envies agent $j$, there exists a good in $j$'s bundle such that $i$ would not envy $j$ after removing that good.\footnote{Note that this good is not actually removed: this is simply a thought experiment used in the definition of EF1.} This property has been studied widely for one-sided preferences, but to our knowledge has not been considered for two-sided markets. We can define EF1 equivalently for two-sided preferences: simply replace ``there exists a good" with ``there exists a match", etc.

The less obvious question is how to adapt EF1 to the repeated setting. In this paper, we assume time is divided into discrete steps, where on each step, some set of matches occur. Each match consists of two agents, one from each side of the market. We would like the cumulative matching after each time step to be EF1. We will also assume that each side of the market has the same number of agents (if not, add ``dummy" agents to the smaller side).

We consider two different versions of this model. In the first version, each time step consists of just a single match, so we are effectively requiring the cumulative matching to be EF1 at every point in time. We call this \emph{EF1-over-time}. 

However, asking for fairness at every point in time may be too strong. Furthermore, in most real-world applications, matches would be happening in parallel anyway. Conversely, EF1-over-time poses no restrictions on how many matches agents receive. In real life, agents often have similar ``capacities" (e.g., most cars have a similar number of seats) and thus  should arguably receive matches at similar rates. These concerns motivate a second version of the model, where on each time step, each agent is matched exactly once (i.e., we select a perfect matching). Thus each time step represents a ``round" of matches (which may happen in parallel). We still require that the cumulative matching is EF1 after time step, and we call this \emph{EF1-over-rounds}.

\subsection{Our results}

We use $v_i(j)$ to denote agent $i$'s value for agent $j$, and assume valuations are additive. We say that valuations are \emph{symmetric} if $v_i(j) = v_j(i)$ for all agents $i,j$, and \emph{binary} if there exists $a \in [0,1)$ such that $v_i(j) \in \{a,1\}$ for all $i,j$.\footnote{If we wish to add dummy agents to one side of the market, the most natural case would be $a=0$, in order to express that dummy agents have no value.} It is worth noting that for symmetric valuations, negative values for $v_i(j)$ are subsumed in the following sense: if the algorithm ever says to match agents $i$ and $j$ where $v_i(j) < 0$, we simply ignore this and never make the match, which gives both agents value 0 for the ``match".\footnote{In order for this argument to be complete when considering EF1-over-rounds, this ``match" must still count as part of the perfect matching for that time step.}


We now describe our results.

\paragraph{EF1-over-rounds for dynamic, symmetric, and binary valuations.} Our main result is that for dynamic, symmetric, and binary valuations, we give an algorithm which both satisfies EF1-over-rounds, and selects a maximum weight matching on each time step (Theorem~\ref{thm:sym-bin}) This holds even when valuations are dynamic.\footnote{We allow valuations to change arbitrarily between time steps. Furthermore, our algorithm does not need to know how valuations will change in response to a given match.}  This shows that for this class of valuations, fairness can be achieved without sacrificing economic efficiency. Our algorithm runs in time $O(n^{2.5})$ per time step.

The class of symmetric and binary valuations is somewhat restricted, but is important to keep several things in mind. First, it is often hard to elicit more complex valuations. Agents can easily answer binary questions such as ``Would you be happy with this match?", but may not be able to provide a real number value for potential matches. Second, the best interpretation of our result (in our opinion) is that agents' preferences are not truly binary, but that our algorithm is guaranteeing EF1 with respect to a \emph{binary projection of the preferences}. That is, ask each agent to label each possible match as ``good" or ``bad", and guarantee EF1 with respect to those preferences. This interpretation is reinforced by the fact that the cumulative matchings computed by our algorithm will be EF1 uniformly across all possible values $a$, and agents can even have different values of $a$. See Section~\ref{sec:sym-bin-setup} for a discussion of this.

Symmetry, however, is a significant assumption on the agents' preferences. There are reasons to believe real-world preferences are largely symmetric: a rider is likely to prefer a driver closer to her, and vice versa. However, a natural question is whether this assumption is necessary.

\paragraph{Counterexamples.} Our next set of results shows that the symmetry assumption is in fact necessary. First, we show that for dynamic binary valuations, EF1-over-rounds alone is impossible (Theorem~\ref{thm:dyn-bin-counter}). Second, for non-dynamic (i.e., valuations do not change over time) binary valuations, it is impossible to satisfy EF1-over-rounds while guaranteeing a maximum weight matching for each time step (Theorem~\ref{thm:bin-max-counter}). These impossibility results suggest that EF1-over-rounds may be too much to ask for in the setting of two-sided repeated matching. However, our counterexamples do not rule out the possibility of EF1-over-time, even for general additive valuations. We leave this as our primary open question.

\paragraph{Beyond symmetric valuations.} Despite this negative result, we show that it is possible to relax the symmetry assumption, at least in the context of EF1-over-time. We show that for $\{0,1\}$ binary valuations\footnote{For this result, we assume that $v_i(j) \in \{0,1\}$ for each agent, as opposed $v_i(j) \in \{a,1\}$ for any $a \in [0,1)$.} with an assumption that we call ``only symmetric cycles", EF1-over-time can be guaranteed (Theorem~\ref{thm:asym-bin}). We formally define ``only symmetric cycles" later, but this assumption is strictly weaker than full symmetry of valuations. 

\paragraph{Beyond binary valuations.} In a similar vein, we show that when one side of the market has two agents, the binary assumption can also be relaxed. Specifically, we give an algorithm which is EF1-over-time for any additive valuations (Theorem~\ref{thm:n=2}).

\subsection{Related work}\label{sec:related}

There are two primary bodies of related work: (one-sided) fair division, and matching markets.

\subsubsection{Fair division}

Fair division has a long history. In fact, the Bible documents Abraham and Lot's use of the cut-and-choose protocol to fairly divide land. The formal study of fair division was started by~\cite{steinhaus48} in 1948, and envy-freeness was proposed in 1958~\cite{Gamow1958a} and further developed by~\cite{Foley1967}. A full overview of the fair division literature is outside the scope of this paper (we refer the interested reader to~\cite{Moulin2004,Brandt2016}), and we discuss only the work most relevant to our own.

There are two main differences between our work and that of traditional fair division. First, we study two-sided preferences instead of one-sided preferences. Second, we study a repeated setting, where we must make an irrevocable decision on each time step; most fair division research considers a ``one-shot" model where all the goods are allocated at once.

We briefly overview some key results in the one-sided one-shot model of fair division for indivisible items\footnote{\emph{Indivisible} items, such as cars, must each go entirely to a single agent. In contrast, \emph{divisible} items, such as a cakes, can be split between multiple players. Fair division studies both of these settings, but the indivisible case is more relevant to our work.}. The EF1 property was proposed for this model by~\cite{Budish2011}. EF1 allocations always exist, and can be computed in polynomial time, even for general combinatorial valuations~\cite{Lipton2004}\footnote{The algorithm of~\cite{Lipton2004} was originally developed with a different property in mind.}. This sweeping positive result for the one-sided one-shot model lies in stark contrast to our negative results for the two-sided repeated model. It was later shown that for additive valuations, maximizing the product of valuations yields an allocation that is both EF1 and Pareto optimal~\cite{Caragiannis2016}. 
\
\paragraph{Envy-freeness up to any good (EFX)} The same paper~\cite{Caragiannis2016} suggested a new fairness notion that is strictly stronger than EF1, which they called EFX\footnote{An allocation is \emph{envy-free up to any good} (EFX) if whenever $i$ envies $j$, removing \emph{any} good from $j$'s bundle eliminates the envy.}. The first formal results regarding EFX allocations were given by~\cite{Plaut2018}. A major breakthrough recently proved the existence of EFX allocations for additive valuations and three agents~\cite{Chaudhury2020}, but despite ongoing effort, the question of existence remains unsolved for more than three agents (or more complex valuations). This is perhaps the most significant open problem in the fair division of indivisible items.

In many contexts (especially for additive valuations), it is common to modify the requirement to be that whenever $i$ envies $j$, removing any good which $i$ values positively from $j$'s bundle is sufficient to eliminate the envy~\cite{Caragiannis2016}.\footnote{The reason this is less common when considering non-additive valuations is that for general combinatorial valuations, it is less clear what ``values positively" means.} Under the latter definition, for $\{0,1\}$ binary valuations, EF1 and EFX coincide. This is because the only positively valued goods are the maximum value goods. In this sense, our positive results for $\{0,1\}$ binary valuations immediately extend to EFX as well.

\paragraph{Repeated fair division.} There are a smattering of recent works studying envy-freeness for one-sided preferences in a repeated (i.e., not one-shot) setting; see~\cite{survey2019} for a short survey. One example is ~\cite{Benade2018}, which focuses on minimizing the maximum envy (i.e., the maximum difference between an agent's value for her own bundle and her value for another agent's bundle) at each time step. Despite the growing interest in repeated one-sided fair allocation, the literature on the analogous two-sided problem remains sparse.

\subsubsection{Matching markets} 

The other relevant field is (bipartite) matching markets\footnote{For a broad overview of this topic, see~\cite{Roth1990}.}. In a \emph{one-to-one} matching market, each agent receives exactly one match. Perhaps the most famous result for one-to-one matching is that of Gale and Shapley, whose algorithm finds a stable matching~\cite{Gale1962}. More relevant to us is the model of \emph{many-to-many matching markets}, where each agent can receive multiple matches; stability has often been the primary criterion in this model as well~\cite{Echenique2006,Konishi2006,Sotomayor1999,Sotomayor2004}. There is also some work on stability in dynamic matching markets~\cite{Damiano2005,Haruvy2007}.

In contrast, fairness in many-to-many matching has received considerably less attention. In fact, Gale and Shapley's algorithm for one-to-one matching is known to compute the stable matching which is the worst possible for one side of the market, and best possible for the other.

\paragraph{Fair ride-hailing.} There is a growing body of work surrounding the ethics of crowdsourced two-sided markets, especially ride-hailing (e.g., Lyft and Uber)~\cite{Bokanyi2018,Bokanyi2019,Calo2017,Fieseler2019,Hannak2017}. We are aware of just two works studying fairness for two-sided markets from an algorithmic perspective: \cite{Wolfson2017} and \cite{Suehr2019}, both of which focus on ride-hailing. The former paper considers ride-\emph{sharing}, where multiple passengers are matched with a single driver. The authors focus on fairness with respect to the savings achieved by each passenger. This paper is primarily theoretical, like ours, but is specific to ride-hailing, unlike ours. The latter paper studies a fairness notion based on the idea that ``spread over time, all drivers should receive benefits proportional to the amount of time they are active in the platform". The model considered in this paper is more general than just ride-hailing, however the paper is primarily experimental, and the experiments are in the ride-hailing setting.

Consequently, we are not aware of any prior work studying algorithms with provable fairness guarantees for repeated two-sided matching markets. In this way, our work can be viewed as simultaneously building on the fair division literature (by considering two-sided preferences) and building on the matching market literature (by studying envy-freeness for repeated two-sided markets).\\

The paper proceeds as follows. Section~\ref{sec:model} describes the formal model. Section~\ref{sec:sym-bin} presents our main result: an algorithm for symmetric and binary valuations such that (1) the sequence of cumulative matchings is EF1-over-rounds, (2) a maximum weight matching is chosen for each time step, and (3) this holds even for dynamic valuations (Theorem~\ref{thm:sym-bin}). In Section~\ref{sec:counter}, we show that this cannot be extended to non-symmetric binary valuations, by showing that without symmetry, (1) and (3) together and (1) and (2) together are both impossible (Theorems~\ref{thm:dyn-bin-counter} and \ref{thm:bin-max-counter}, respectively). Section~\ref{sec:asym-bin} presents our result for binary valuations with only symmetric cycles (Theorem~\ref{thm:asym-bin}), and Section~\ref{sec:n=2} presents our result for the case when one side of the market has only two agents (Theorem~\ref{thm:n=2}).

\section{Model}\label{sec:model}

Let $N$ and $M$ be two sets of agents. We assume that $|N| = |M| = n$; if this is not the case, we can add ``dummy" agents (i.e., agents $i$ such that $v_i(j) = v_j(i) = 0$ for all $j$) to the smaller side of the market until both sides have the same number of agents. We will typically use odd numbers for the elements of $N$ and even numbers for the elements of $M$, i.e., $N = \{1, 3,\dots, 2n-1\}$ and $M = \{2,4,\dots, 2n\}$. A \emph{matching} $X$ assigns a multiset of agents in $N$ to each agent in $M$, and a multiset of agents in $M$ to each agent in $N$. For each $i \in N\cup M$, we will use $X_i$ to denote agent $i$'s \emph{bundle}, i.e., the multiset of agents she is matched to. Throughout the paper, we will use standard set notation for operations on the multisets $X_i$. For example, $X_i \cup \{j\}$ increments the \emph{multiplicity} of $j$ in $X_i$, i.e., the number of times $j$ occurs in $X_i$. In order for $X$ to be a valid matching, the multiplicity of $j$ in $X_i$ must be equal to the multiplicity of $i$ in $X_j$ for each $i \in N$, $j \in M$.

Each agent $i$ also has a \emph{valuation function} $v_i$, which assigns a real number to each possible bundle she might receive. We will use $\V$ to denote the \emph{valuation profile} which assigns valuation $v_i$ to agent $i$. We say that $v_i$ is \emph{additive} if for any bundle $X_i$,
\[
v_i(X) = \sum_{j \in X} v_i(\{j\})
\]
Since $X$ is a multiset, the sum over $j\in X$ includes each $j$ a number of times equal to its multiplicity. For example, if $X = \{j,j\}$ and $v_i(\{j\}) = 1$, then $v_i(X) = 2$. With slight abuse of notation, we will write $v_i(\{j\}) = v_i(j)$. We say that a valuation $v_i$ is \emph{binary} if there exists $a \in [0,1)$ such that $v_i(j) \in \{a,1\}$ for all $i,j \in N$ or $i,j \in M$. We say that a valuation profile $\V$ is \emph{symmetric} if $v_i(j) = v_j(i)$ for all $i\in N, j\in M$.

We say that $i$ \emph{envies} $j$ under $X$ if $v_i(X_i) < v_i(X_j)$. We only consider envy within the same side of the market: it is unclear what it would mean for some $i \in N$ to envy $j \in M$. We can express this by setting $v_i(j) = 0$ for $i,j \in N$ or $i,j \in M$.
\begin{definition}\label{def:ef1}
A matching $X$ is \emph{envy-free up to one match} (EF1) if whenever $i$ envies $j$, there exists $k \in X_j$ such that $v_i(X_i) \ge v_i(X_j\setminus\{k\})$.
\end{definition}

Note that the set subtraction $X_j\setminus\{k\}$ decreases the multiplicity of $k$ by 1; it does not remove $k$ altogether. Also, we will say that $X$ is EF1 with respect to $N$ (resp., $M$) if the above holds for every pair $i,j\in N$ (resp., $i,j \in M$).

One important tool we will use is the \emph{envy graph}:

\begin{definition}
The \emph{envy graph} of a matching $X$ is a graph with a vertex for each agent, and a directed edge from agent $i$ to agent $j$ if $i$ envies $j$ under $X$.
\end{definition}

We will especially be interested in cycles in the envy graph, and will use the terms ``cycle in the envy graph" and ``envy cycle" interchangeably.

\subsection{Repeated matching}

We consider a repeated setting, where on each time step $t$, some set of matches occur. Each ``match" (alternatively, \emph{pairing}) consists of one agent in $N$ and one agent in $M$.

Let $x^t$ denote the set of matches which occur at time $t$. Each agent will receive at most one match per time step. If agent $i$ is matched to an agent $j$ at time $t$, let $x^t_i = \{j\}$; otherwise, let $x^t_i = \emptyset$. For an infinite sequence $x^1, x^2, x^3\dots$, let $X^t$ denote the cumulative matching up to and including time $t$. Formally, for each $i \in N \cup M$,
\[
X_i^t = \begin{cases}
\emptyset & \text{if } t =0\\
X_i^{t-1}\cup x^t_i & \text{if } t>0
\end{cases}
\]
In words, $X_i^t$ is the set of matches $i$ has received up to and including time $t$.

Our main result holds even when valuations are allowed to vary over time. Specifically, a \emph{dynamic} valuation $v_i$ will have a value $v_i^t(j)$ for each agent $j$ on each time step $t$ (as before, we write $v_i^t(j) = 0$ for $i,j \in N$ or $i,j\in M$). A profile of dynamic valuations is symmetric if $v_i^t(j) = v_j^t(i)$ for all $i,j,t$. For a pair of agents $i,j$ (with $i=j$ allowed), $v_i(X_j^t)$ is given by
\[
v_i(X_j^t) = \sum_{t' = 1}^t v_i^{t'}(x^{t'}_j)
\]
where $v_i^{t'}(\emptyset) = 0$. That is, $i$'s value for a bundle $X_j$ is as if $i$ had received exactly those matches at exactly those times. It is important for this definition to include both $i = j$ and $i \ne j$, so that we can evaluate envy between agents.

We make no assumptions on how valuations change between time steps: they can even change adversarially, since our algorithm will not use any knowledge about future valuations when making matching decisions.

We consider two definitions of EF1 in the repeated matching setting. In both cases, we require the cumulative matching at the end of each time step to be EF1. The difference is that for EF1-over-time, each time step consists a single match, and for EF1-over-rounds, each time step consists of a ``round" of matches where all agents receive exactly one match (i.e., a perfect matching between $N$ and $M$).

\begin{definition}\label{def:ef1-over-time}
The sequence $\X = X^0, X^1,X^2\dots$ is \emph{EF1-over-time} if for all $t \ge 0$, each $x^t$ contains a single match, and $X^t$ is EF1.
\end{definition}

\begin{definition}\label{def:ef1-over-rounds}
The sequence $\X = X^0, X^1,X^2\dots$ is \emph{EF1-over-rounds} if for all $t \ge 0$, $x^t$ is a perfect matching, and $X^t$ is EF1.
\end{definition}

Formally, these notions are incomparable: EF1-over-time has a stronger fairness requirement (the cumulative matching should be EF1 after every match, not just after every round of matches), but does not require agents to receive the same number of matches. However, EF1-over-rounds does imply EF2-over-time (where we may remove two matches in order to eliminate the envy): expand each ``round" into $n$ time steps, each containing one match, in an arbitrary order. We know that at the end of each round of $n$ time steps, the cumulative matching is EF1. Within each round, each agent only gains one additional match, and we can always remove that match to return to an EF1 state.

Our goal will be to show the existence of (and efficiently compute) a sequence $x^1, x^2, x^3\dots$ such that the induced sequence $\X$ is EF1-over-time and/or EF1-over-rounds. For brevity, if an algorithm is guaranteed to produce a sequence $\X$ that is EF1-over-time (resp., EF1-over-rounds), we simply say that the algorithm is EF1-over-time (resp., EF1-over-rounds).

\section{EF1 for dynamic, symmetric, and binary valuations}\label{sec:sym-bin}

In this section, we consider binary and symmetric valuations that may change over time. For this class of valuations, we give a polynomial-time algorithm that produces a sequence which is EF1-over-rounds, and chooses a maximum weight matching for each time step. This leads to the following theorem:

\begin{restatable}{theorem}{thmSymBin}
\label{thm:sym-bin}
For dynamic, binary, and symmetric valuations, Algorithm~\ref{alg:sym-bin} is EF1-over-rounds, and the matching $x^t$ for each time step $t$ is a maximum weight matching (with respect to the valuations on that time step). Furthermore, the algorithm runs in time $O(n^{2.5})$ per time step.
\end{restatable}

The section is organized as follows. In Section~\ref{sec:sym-bin-setup}, we define our algorithm and discuss some basic properties. Section~\ref{sec:roadmap} gives some intuition behind the algorithm and proof. The rest of the section (Sections~\ref{sec:sym-bin-max} -- \ref{sec:sym-bin-final}) is devoted to proving Theorem~\ref{thm:sym-bin}.

\subsection{Algorithm setup}\label{sec:sym-bin-setup}

Before we discuss the algorithm, we need the following definition, which will imply EF1 (Lemma~\ref{lem:envy-bounded}):

\begin{definition}\label{def:envy-bounded}
We say that a pair of agents $(i,j)$ is \emph{$c$-envy-bounded} if $v_i(X_j) - v_i(X_i) \le c$, and we say a matching $X$ is $c$-envy-bounded if every pair $(i,j)$ is $c$-envy-bounded.
\end{definition}

A quick note: recall that our goal is to choose a sequence of pairings $x^1, x^2\dots$, and that these pairings fully specify the sequence of cumulative matchings $\X$. Consequently, when giving pseudocode for our algorithms (throughout the paper), we do not explicitly update $\X$: we assume that whenever some $x^t$ is changed, every $X^{t'}$ for $t' \ge t$ is automatically updated. We feel that this leads to more concise and intuitive pseudocode.

Algorithm~\ref{alg:sym-bin} is very simple. For each time step $t$, we initialize $x^t$ to be an arbitrary maximum weight matching for the current valuations, and make changes to this matching until we are satisfied. Specifically, while there exist agents $i,j$ such that $(i,j)$ is not $(1-a)$-envy-bounded in the cumulative matching, we swap their matches in $x^t$. When no such pair of agents exists, we exit the while loop and confirm the matches. Throughout all of our algorithms, we will use the function \texttt{MakeMatch} to indicate that we are confirming the matches in $x^t$.

\begin{algorithm*}[htb]
\centering
\begin{algorithmic}[1]
\Function{EF1Matching}{$N, M, \V$}
\ForAll{$t\in \mathbb{N}_{\ge 0}$}
    \State $\{x^t\} \gets \texttt{MaxWeightMatching}(N,M, \V)$
    \While{$\exists$ agents $i,j$ s.t. $v_i(X_j^t) - v_i(X_i^t) > 1-a$}
    	\State $(x^t_i, x^t_j) \gets (x^t_j, x^t_i)$
    \EndWhile
\State \texttt{MakeMatch}($x^t$)
  \EndFor
\EndFunction
\end{algorithmic}
\caption{An EF1-over-rounds algorithm for agents with dynamic, symmetric, and binary valuations.}
\label{alg:sym-bin}
\end{algorithm*}

It is important to note that the algorithm is \emph{not} going back in time and changing pairings already made: once a pairing is confirmed with \texttt{MakeMatch}, it is never changed. The algorithm starts with a tentative matching, and changes \emph{tentative} matches until it is satisfied for the current time step (see Figure~\ref{fig:sym-bin-swap}), at which point the matches are confirmed with \texttt{MakeMatch}. The algorithm then proceeds to the next time step and never changes pairings from previous time steps. Note also that the algorithm uses no information about valuations for future time steps.

\begin{figure}[htb]
\centering
\begin{tikzpicture}[thick,every node/.style={draw,circle,minimum size=0 mm, node distance=1.7 cm, font=\footnotesize,fill=blue!20}]
  \node (1) at (0,0) {$1$};
  \node (3) [right of=1] {$3$};
  \node (2) [below of=1] {$2$};
  \node (4) [right of=2] {$4$};
  \node (5) at (5,0) {$1$};
  \node (7) [right of=5] {$3$}; 
  \node (6) [below of=5] {$2$}; 
  \node (8) [right of=6] {$4$};
  \node[draw=none,fill=none] (9) [below right of=3] {};
  \node[draw=none,fill=none] (10) [below left of=5] {};
  \node[draw=none,fill=none] (N) [left of=1] {$N:$};
  \node[draw=none,fill=none] (M) [left of=2] {$M:$};
  
  \draw[->] (9.north west) -- (10.north east);
  
\path[->,color=blue, line width=.5mm]
	(3) edge[bend right] (1)
	;
  
\path[every node/.style={font=\small}]
	(1) edge node[left] {} (2)
	(3) edge node[right] {} (4)
	(5) edge node[left] {} (8)
	(7) edge node[right] {} (6)		
	;
\end{tikzpicture}
\caption{A hypothetical swap performed by Algorithm~\ref{alg:sym-bin}. On the left we see a tentative perfect matching: $(\{1,2\}, \{3,4\})$. The blue arrow indicates that if this matching were to be chosen, the pair $(3,1)$ would not be $(1-a)$-envy-bounded. Thus agents 1 and 3 swap their (tentative) matches, and the new tentative matching is $(\{1,4\}, \{3,2\})$. The matching $(\{1,2\}, \{3,4\})$ is never confirmed by \texttt{MakeMatch}: it is merely a stepping stone in the process of computing the eventual matches to be chosen for this time step. For the case of more than four agents, this process could repeat (although not indefinitely; see Lemma~\ref{lem:termination}).}
\label{fig:sym-bin-swap}
\end{figure}
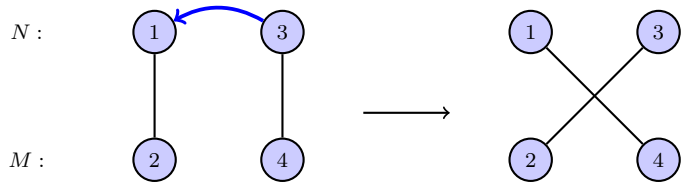

Our central correctness lemma will be the following:

 \begin{restatable}{lemma}{lemFullStage}
 \label{lem:full-stage}
Let $t \ge 1$ be any time step, and suppose that $X^{t-1}$ is $(1-a)$-envy-bounded and has no envy cycles. Then $X^t$ is $(1-a)$-envy-bounded and has no envy cycles. Furthermore, the chosen matching $x^t$ is a maximum weight matching (with respect to the valuations on that time step).
 \end{restatable}

Before diving into the proof of Lemma~\ref{lem:full-stage} (and the runtime analysis), we briefly show that $(1-a)$-envy-boundedness will actually give us the result we want:
 
\begin{lemma}\label{lem:envy-bounded}
Suppose valuations are binary, and suppose $X$ is $(1-a)$-envy-bounded. Then $X$ is EF1.
\end{lemma}

\begin{proof}
Suppose $i$ envies $j$ under $X^t$. If $v_i(X_j^t) = a|X_j^t|$, then $v_i(X_i^t) \ge v_i(X_j^t)$, which contradicts $i$ envying $j$. Thus $v_i(X_j^t) \ge 1 + a(|X_j^t| - 1)$. Thus there exists $k \in X_j^t$ such that $v_i(X_j^t\setminus\{k\}) = v_i(X_j^t) - 1$. Therefore $v_i(X_i) - v_i(X_j^t\setminus\{k\}) \ge 1 + (a-1) = a \ge 0$, which proves the claim.
\end{proof}

\subsubsection{The role of $a$}

Before diving into the main proof, we briefly discuss the role of $a$. For Theroem~\ref{thm:sym-bin}, we assume that there exists $a \in [0,1)$ such that $v_i(j) \in \{a,1\}$ for all $i,j \in N$ or $i,j \in M$. For ease of notation, we assume that all agents have the same value of $a$, but this is in fact not necessary. In fact, Algorithm~\ref{alg:sym-bin} will be EF1-over-rounds simultaneously for all values of $a$.

\begin{lemma}\label{lem:envy-bounded-uniform}
Assume $(i,j)$ is $(1-a)$-envy-bounded and that $|X_i^t| = |X_j^t|$. Let $a' \in [0,1)$, and define a new valuation $v'_i$ such that $v'_i(k) = a'$ whenever $v_i(k) = a$, and $v'_i(k) = 1$ otherwise. Then $(i,j)$ is $(1-a')$-envy-bounded with respect to $v'_i$.
\end{lemma}

\begin{proof}
For a bundle $X$, let $\kappa_i(X)$ be the number of matches in $X$ for which agent $i$ has value 1. Then we have $v_i(X) = \kappa_i(X) + a(|X| - \kappa_i(X)) = a|X| - (1-a)\kappa_i(X)$. Therefore,
\begin{align*}
v_i(X_j^t) - v_i(X_i^t) =&\ \big(a|X_j^t| - (1-a)\kappa_i(X_j^t)\big) - \big(a|X_i^t| - (1-a)\kappa_i(X_i^t)\big)\\
=&\ (1-a)(\kappa_i(X_i^t) - \kappa_i(X_j^t))
\end{align*}
Since $(i,j)$ is $(1-a)$-envy-bounded with respect to $v_i$, we have $v_i(X_j^t) - v_i(X_i^t) \le 1-a$, which implies $\kappa_i(X_i^t) - \kappa_i(X_j^t) \le 1$. Therefore
\[
v'_i(X_j^t) - v'_i(X_i^t) = (1-a')(\kappa_i(X_i^t) - \kappa_i(X_j^t)) \le (1-a')
\]
Thus $(i,j)$ is $(1-a')$-envy-bounded with respect to $v'_i$, as required.
\end{proof}

Note that the assumption of $|X_i^t| = |X_j^t|$ is always satisfied when working with EF1-over-rounds, since we will match every agent once on each time step. Therefore we can actually just choose an arbitrary value of $a \in [0,1)$ and run Algorithm~\ref{alg:sym-bin}. Lemma~\ref{lem:envy-bounded-uniform} implies that the resulting sequence of matchings will be EF1-over-rounds simultaneously for all values of $a$, even if different agents have different values of $a$. That said, if we need to include dummy agents in order to equalize the sizes of $N$ and $M$, $a = 0$ probably makes the most sense.

\subsection{Proof intuition}\label{sec:roadmap}

\textbf{Termination of the while loop.} The while loop condition ensures that once the while loop terminates, the cumulative matching is $(1-a)$-envy-bounded. Perhaps the first concern is whether the while loop for a given time step actually terminates. The key property we will need is the absence of envy cycles.

Consider an arbitrary iteration of the while loop, and define $i,j$ as in Algorithm~\ref{alg:sym-bin}. We will show that $v_i(X_j^t) - v_i(X_i^t) > 1-a$ implies two things: (1) $i$ envied $j$ before their most recent matches, and (2) $i$ prefers $j$'s proposed match to her own. When $i$ and $j$ swap their proposed matches, because of (2), $i$'s utility increases. The swap may hurt $j$, but because $i$ precedes $j$ in the envy DAG (here we crucially use the absence of envy cycles), we can show that only so many swaps occur on this side of the market before termination.

However, we also must consider what a swap does to the other side of the market. We show that any swap does not hurt the other side of the market, and the crucial property for this is symmetry of valuations. Let $k$ and $\ell$ denote $i$ and $j$'s most recent matches before the swap, respectively. Since $i$ prefers $\ell$ to $k$, and valuations are binary, we have $v_i(\ell) = 1$ and $v_i(k) = a$. Then by symmetry, $v_\ell(i) = 1$ and $v_k(i) = a$. Thus after the swap, $\ell$ is happy (because she is matched to $i$), and before the swap, $k$ is unhappy (because she is matched to $i$). Thus the swap cannot hurt either $k$ or $\ell$. This is important, because this means that this swap does not affect the progress we are making with respect to the envy DAG on the other side of the market.

\textbf{Maximum weight matching for each time step.} The above reasoning also allows us to show that Algorithm~\ref{alg:sym-bin} produces a maximum weight matching for each time step. The proof of this is inductive, using the fact that the matching is initially maximum weight by construction, and that no swap decreases the weight of the matching. Define $i,j,k,\ell$ as above. We know that $i$ dislikes her match before the swap, and likes her match after the swap. We also know that $k$ dislikes her match before the swap, and $\ell$ likes her match after the swap. We will claim that out of $i,j,k,\ell$, at most two are happy before the swap, and at least two are happy after the swap. This allows us to show that the weight of the matching does not decrease (and since it is initially maximum weight, does not increase).

\textbf{Absence of envy cycles.} The above arguments only work if there are no envy cycles at the beginning of each time step $t$. The key property we need for this is that $x^{t-1}$ (i.e., the matching from the previous time step) is a maximum weight matching. We show that the existence of an envy cycle $C$ would imply the presence of a higher weight matching, which is a contradiction.

Consider a 2-cycle between agents $i$ and $j$. If we assume inductively that there were no envy cycles at the end of the previous time step, at least one of these edges must be new; say that $i$ did not previously envy $j$. This implies that $i$ strictly prefers $j$'s most recent match to her own. Furthermore, we show that $j$ weakly prefers $i$'s most recent match to her own: otherwise she would not envy $i$. Thus if agents $i$ and $j$ swapped matches, $i$ would be strictly happier and $j$ would be weakly happier. Using the same reasoning as above, such a swap does not hurt the other side of the market, so the weight of the matching has increased. The same intuition holds for envy cycles with more than two agents.

\textbf{Putting it all together.} The final proof does a simultaneous induction of several properties. Specifically, for each time step $t$, if the cumulative matching $X^t$ (1) is $(1-a)$-envy-bounded, and (2) has no envy cycles, then both (1) and (2) hold for $X^{t+1}$. Since $(1-a)$-envy-boundedness implies EF1, this will show that the algorithm is EF1-over-rounds.

Keeping this intuition in mind, the proof proceeds as follows. Sections~\ref{sec:sym-bin-max} shows that if $X^{t-1}$ is $(1-a)$-envy-bounded, then $x^t$ will be a maximum weight matching (Lemma~\ref{lem:max-matching}). Section~\ref{sec:termination} shows that the while loop terminate after $O(n^2)$ iterations (Lemma~\ref{lem:termination}), and Section~\ref{sec:envy-cycles} shows that no envy cycles are created (Lemma~\ref{lem:no-ec}). Section~\ref{sec:sym-bin-final} puts all of this together and proves Theorem~\ref{thm:sym-bin}.

\subsection{We end up with a maximum weight matching}\label{sec:sym-bin-max}

Next, we provide some brief definitions.

\begin{definition}\label{def:stage-matching}
The weight of a perfect matching $x$ is given by
\[
W^t(x) = \frac{1}{2}\sum_{\{i,j\} \in x} v_i^t(j) + v_j^t(i)
\]
\end{definition}

When valuations are symmetric, we have $v_i^t(j) + v_j^t(i) = 2v_i^t(j)$, but we feel it is better to give a more general definition that does not depend on assumptions on agent preferences.

Thus \texttt{MaxWeightMatching} simply selects a matching with maximum weight. Since our graph is bipartite, this can be done by standard maximum matching algorithms in time $O(n^{2.5})$. Note that since valuations are dynamic (i.e., can change over time), different time steps may have different maximum weight matchings. Thus we say that $x^t$ is a maximum weight matching, we mean with respect to the valuations on time step $t$.

Consider an arbitrary perfect matching $x$. Since valuations are symmetric, each match is either a ``good match" (i.e., both agents have value 1 for the match) or a ``bad match" (i.e., both agents have value $a$ for the match). Let $E_t(x)$ be the set of ``good matches", i.e.,
\[
E_t(x) = \{\{i,j\} \in x: v_i^t(j) = v_j^t(i) = 1\}
\]
Note that $E_t(x)$ is a set of (undirected) edges between $N$ and $M$ such that each agent participates in at most one edge. Every edge in $E_t(x)$ contributes 1 to $W^t(x)$, and every edge in $x \setminus E_t(x)$ contributes $a$. Thus the size of $E_t(x)$ fully determines $W^t(x)$; in particular,
\begin{align}
W^t(x) = |E_t(x)| + a(|x|-|E_t(x)|) \label{eq:weight}
\end{align}

\begin{lemma}\label{lem:swap-property}
Suppose that $X^{t-1}$ is $(1-a)$-envy-bounded. For an arbitrary iteration of the while loop during time step $t$, let $i,j$ be as defined in Algorithm~\ref{alg:sym-bin}, and let $\{i,k\}$, $\{j,\ell\}$ be the matches before the swap. Then $v_i^t(k) = a$ and $v_i^t(\ell) = 1$. Also, $v_i(X^{t-1}_i) < v_i(X^{t-1}_j)$.
\end{lemma}

\begin{proof}
Before the swap, we have $X_i^t = X_i^{t-1} \cup \{k\}$ and $X_j^t = X_j^{t-1} \cup \{\ell\}$. By construction, we know that before the swap, $v_i(X_j^t) - v_i(X_i^t) > 1-a$. Since we assumed that $X^{t-1}$ is $(1-a)$-envy-bounded, we have $v_i(X_j^t\setminus\{\ell\}) - v_i(X_i^t\setminus\{k\}) \le 1-a$. Because valuations are additive, we have
\begin{align*}
1-a \ge&\ v_i(X_j^t\setminus\{\ell\}) - v_i(X_i^t\setminus\{k\})\\
=&\ v_i(X_j^t) - v_i(X_i^t) + v_i^t(k) - v_i^t(\ell)\\
>&\ 1-a + v_i^t(k) - v_i^t(\ell)\\
\end{align*}
Therefore $v_i^t(\ell) > v_i^t(k)$. Since valuations are binary, this implies that $v_i^t(\ell) = 1$ and $v_i^t(k) = a$. This also gives us $v_i(X_j^t\setminus\{\ell\}) - v_i(X_i^t\setminus\{k\}) = v_i(X^{t-1}_j) - v_i(X^{t-1}_i) > 0$, as required.
\end{proof}

\begin{lemma}\label{lem:e-one-round}
Suppose that $X^{t-1}$ is $(1-a)$-envy-bounded. Then on each iteration of the while loop during time step $t$, $|E_t(x^t)|$ does not decrease.
\end{lemma}

\begin{proof}
 Let $i,j$ be as defined in Algorithm~\ref{alg:sym-bin}, and let $\{i,k\}$ and $\{j,\ell\}$ be the matches before the swap. Since all agents outside of $\{i,j,k,\ell\}$ are unaffected, we need only consider those four agents. 
 
 By Lemma~\ref{lem:swap-property}, $v_i^t(\ell) = 1$ and $v_i^t(k) = a$. Thus the edge $\{i,k\}$ cannot be in $E_t(x^t)$, so the set $\{i,j,k,\ell\}$ contributes at most 1 edge to $E$ (both before and after the swap).

Now consider $E_t(x^t)$ after $i$ and $j$ swap matches. Since $v_i^t(\ell) = 1$, we have $\{i,\ell\} \in E_t(x^t)$. Thus after the swap, the set $\{i,j,k,\ell\}$ contributes at least one edge to $E_t(x^t)$. Since the matches of agents outside of $\{i,j,k,\ell\}$ are unchanged, it follows that the size of $E_t(x^t)$ has not decreased.
\end{proof}

We are now ready to show that $x^t$ remains a maximum weight matching throughout the duration of the while loop. Note that we have not yet shown that the while loop is guaranteed to terminate; we will do that next.

\begin{lemma}\label{lem:max-matching}
Suppose that $X^{t-1}$ is $(1-a)$-envy-bounded. Then after every iteration of the while loop, $x^t$ is a maximum weight matching.
\end{lemma}

\begin{proof}
We proceed by induction. Before we enter the while loop, $x^t$ is a maximum weight matching by construction. For the inductive step, consider an arbitrary iteration of the while loop, and assume that $x^t$ has maximum weight at the beginning of this iteration. Lemma~\ref{lem:e-one-round} implies that $|E_t(x^t)|$ does not decrease during this iteration, so by Equation~\ref{eq:weight}, $W^t(x^t)$ does not decrease. Since $x^t$ had maximum weight before the iteration, $W^t(x^t)$ cannot increase. We conclude that $x^t$ has maximum weight at the end of this iteration.
\end{proof}

Once we show that the while loop for each time step is guaranteed to terminate, Lemma~\ref{lem:max-matching} will imply that we end up with a maximum weight matching for each time step $t$.

\subsection{The while loop terminates}\label{sec:termination}

\begin{lemma}\label{lem:termination}
Suppose that $X^{t-1}$ is $(1-a)$-envy-bounded and has no envy cycles. Then the while loop on time step $t$ has at most $2n^2$ iterations before terminating.
\end{lemma}

\begin{proof}
Consider a given iteration of the while loop on time step $t$. As before, let $i,j$ be as defined in Algorithm~\ref{alg:sym-bin}, and let $\{i,k\}$ and $\{j,\ell\}$ be the matches before the swap. 

This iteration changes the matches of exactly four agents: $i,j,k,$ and $\ell$. By Lemma~\ref{lem:swap-property}, $v_i^t(\ell) = 1$ and $v_i^t(k) = a$. Thus $\{i,k\}$ cannot be in $E_t(x^t)$ before the swap, and we know that after the swap, we have $\{i,\ell\} \in E_t(x^t)$. Furthermore, since $|E_t(x^t)|$ cannot change (that would change the weight of the matching, which would contradict  Lemma~\ref{lem:max-matching}), we know that $\{j,\ell\} \in E_t(x^t)$ before the swap, and $\{j,k\} \not\in E_t(x^t)$ after the swap.

We know that $i$ envies $j$ under $X^{t-1}$. Thus on each iteration of the while loop, a single edge $\{j,\ell\} \in E_t(x^t)$ changes to $\{i,\ell\}$ such that $i$ envies $j$ under $X^{t-1}$, and the rest of $E_t(x^t)$ is unchanged. Without loss of generality, assume $i,j \in N$; then the endpoint of the edge in $N$ has changed, but the endpoint of the edge in $M$ (i.e., $\ell$) has not changed.

We claim that there are at most $n^2$ iterations where an edge's endpoint in $N$ changes. Suppose not: since there are at most $n$ edges in $E_t(x^t)$, one of these edges must have changed its endpoint in $N$ more than $n$ times. Thus the endpoint must have cycled, which contradicts our assumption that $X^{t-1}$ has no envy cycles. Thus there are at most $n^2$ iterations where an edge's endpoint in $N$ changes, and by symmetry, there are at most $n^2$ where an edge's endpoint in $M$ changes. Since an edge changes on every iteration of the while loop, we conclude that there are at most $2n^2$ iterations of the while loop.
\end{proof}

\subsection{We never create an envy cycle}\label{sec:envy-cycles}

Given a cycle $C = (1,2\dots |C|)$ and an agent $i \in C$, let $\suc(i) = i\bmod|C| + 1$. Recall that for each agent $i$, $x^t_i$ denotes the unique agent $j$ to whom $i$ is matched during time step $t$.

\begin{lemma}\label{lem:ec-property}
Fix a time step $t$ and assume that $X^{t-1}$ has no envy cycles. Suppose that $X^t$ is $(1-a)$-envy-bounded, but has an envy cycle $C$. Then for all $i \in C$, $v_i^t(x_{\suc(i)}^t) \ge v_i^t(x_i^t)$, and there exists $i \in C$ for whom the inequality is strict.
\end{lemma}

\begin{proof}
Let $C = (1, 2, \dots |C|)$ be an envy cycle in $X^t$, i.e., agent $i$ envies agent $\suc(i)$ under $X^t$ for all $1 \le i \le |C|$. Formally, $v_i(X_i^t) < v_i\big(X_{\suc(i)}^t\big)$ for all $i \in C$. Since the envy graph of $X^{t-1}$ had no cycles, at least one of these edges must be new: there must exist $i$ such that $v_i(X_i^{t-1}) \ge v_i\big(X_{\suc(i)}^{t-1}\big)$. Since $v_i$ is additive, this implies that $v_i^t(x_{\suc(i)}^t) > v_i^t(x_i^t)$.

It remains to show that $v_i^t(x_{\suc(i)}^t) \ge v_i^t(x_i^t)$ for all $i \in C$. Suppose there exists $i \in C$ such that $v_i^t(x_i^t) > v_i^t(x_{\suc(i)}^t)$. Since $v_i$ is binary by assumption, we have $v_i^t(x_i^t) - v_i^t(x_{\suc(i)}^t) = 1-a$. Also, since $X^t$ is $(1-a)$-envy-bounded, we have $v_i(X_{\suc(i)}^t) - v_i(X_i^t) \le 1-a$. Therefore
\begin{align*}
v_i(X_i^t) - v_i(X_{\suc(i)}^t) =&\ v_i(X_i^{t-1}) - v_i(X_{\suc(i)}^{t-1}) + v_i^t(x_i^t) - v_i^t(x_{\suc(i)}^t)\\
\ge&\ (a-1) + (1-a)\\
\ge& 0
\end{align*}
But this implies that $j$ does not envy $\suc(i)$ in $X^t$, which contradicts $C$ being an envy cycle. Therefore $v_i^t(x_{\suc(i)}^t) \ge v_i^t(x_i^t)$ for all $i \in C$, and the inequality is strict for at least one $i \in C$.
\end{proof}

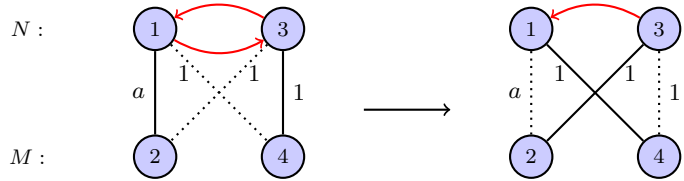
\begin{figure}[ht!bp]
\centering
\begin{tikzpicture}[thick,every node/.style={draw,circle,minimum size=0 mm, node distance=1.7 cm, font=\footnotesize,fill=blue!20}]
  \node (1) at (0,0) {$1$};
  \node (3) [right of=1] {$3$};
  \node (2) [below of=1] {$2$};
  \node (4) [right of=2] {$4$};
  \node (5) at (5,0) {$1$};
  \node (7) [right of=5] {$3$}; 
  \node (6) [below of=5] {$2$}; 
  \node (8) [right of=6] {$4$};
  \node[draw=none,fill=none] (9) [below right of=3] {};
  \node[draw=none,fill=none] (10) [below left of=5] {};
  \node[draw=none,fill=none] (N) [left of=1] {$N:$};
  \node[draw=none,fill=none] (M) [left of=2] {$M:$};
  
  \draw[->] (9.north west) -- (10.north east);
  
\path[->,color=red]
	(1) edge[bend right] (3)
	(3) edge[bend right] (1)
	(7) edge[bend right] (5)
	;
  
\path[every node/.style={font=\small}]
	(1) edge node[left] {$a$} (2)
	(3) edge node[right] {$1$} (4)
	(5) edge node[left, pos=.3] {$1$} (8)
	(7) edge node[right, pos=.3] {$1$} (6)		
	;
	
\path[dotted, every node/.style={font=\small}]
	(1) edge node[left, pos=.3] {$1$} (4)
	(3) edge node[right, pos=.3] {$1$} (2)
	(7) edge node[right] {$1$} (8)	
	(5) edge node[left] {$a$} (6)
	;
  
\end{tikzpicture}
\caption{An example of how the existence of an envy cycle implies a higher weight matching. The figure shows an instance with $N = \{1,3\}$ and $M = \{2,4\}$, where the edge between $i$ and $j$ is labeled with the value of $v_i(j)$ (which is equal to $v_j(i)$, by symmetry). Directed red edges denote envy. On the left hand side, we see a possible perfect matching $x^t$ (solid black edges denote edges in the matching) which could result in an envy cycle between agents 1 and 3. As per Lemma~\ref{lem:ec-property}, agent 1 strictly prefers agent 3's match to her own, and agent 3 weakly prefers agent 1's match to her own. The right hand side shows a different matching with higher weight, obtained by permuting the matches of agents 1 and 3.}
\label{fig:envy-cycle}
\end{figure}

\begin{lemma}\label{lem:no-ec}
Fix a time step $t$. Suppose that $X^{t-1}$ has no envy cycles, and $X^t$ is $(1-a)$-envy-bounded. Then $X^t$ has no envy cycles.
\end{lemma}

\begin{proof}
Lemma~\ref{lem:max-matching} implies that the final $x^t$ is a maximum weight matching. We will show that the existence of an envy cycle in $X^t$ would imply the existence of a matching with higher weight, which is impossible.

Let $C = (1,2\dots |C|)$ be an envy cycle in $X^t$. We define a new perfect matching $y$ which will have higher weight than $x^t$. Define $y$ to be identical to $x^t$, except that we permute the matches of agents in $C$. Let $\suc(i)$ and $\pred(i)$ denote $j$'s successor and predecessor with respect to $C$; then
\[
y_i = \begin{cases}
x_{\suc(i)}^t & \text{if } i \in C\\
x_i^t & \text{if } i\not \in C
\end{cases}
\]
where $y_i$ is the match agent $i$ receives in the matching $y$.

By Lemma~\ref{lem:ec-property}, $v_i^t(x_{\suc(i)}^t) \ge v_i^t(x_i^t)$ for all $i \in C$, and there exists $i \in C$ for whom the inequality is strict. Thus $v_i^t(y_i^t) \ge v_i^t(x_i^t)$ for all $i \in C$, and there exists $i \in C$ for whom the inequality is strict. For all $i\not \in C$, $v_i^t(y_i) = v_i^t(x^t_i)$. Without loss of generality, assume $C \subseteq N$; then we have $\sum_{i \in N} v_i^t(y_i) > \sum_{i \in N} v_i^t(x_i^t)$. By symmetry of valuations, for any matching $x$, we have $\sum_{i \in M} v_i^t(x_i) = \sum_{i \in N} v_i^t(x_i)$. Therefore $\sum_{i \in N\cup M} v_i^t(y_i) > \sum_{i \in N\cup M} v_i^t(x_i^t)$, which is equivalent to $\sum_{\{i,j\} \in y} (v_i^t(j) + v_j^t(i)) > \sum_{\{i,j\} \in x^t} (v_i^t(j) + v_j^t(i))$>

Therefore $W^t(y) > W^t(x^t)$, which contradicts the fact that $x^t$ is a maximum weight matching. We conclude that $X^t$ has no envy cycles.
\end{proof}

\subsection{Putting it all together}\label{sec:sym-bin-final}

\lemFullStage*

\begin{proof}
First, since $X^{t-1}$ is $(1-a)$-envy-bounded and has no cycles, Lemma~\ref{lem:termination} implies that the while loop for time step $t$ terminates. Thus $X^t$ is well-defined. Also, Lemma~\ref{lem:max-matching} immediately implies that the matching $x^t$ chosen is a maximum weight matching. Next, $X^t$ must be $(1-a)$-envy-bounded, or we would not have exited the while loop for time step $t$. Finally, Lemma~\ref{lem:no-ec} implies that $X^t$ has no envy cycles.
\end{proof}

\thmSymBin*

\begin{proof}
By induction, Lemma~\ref{lem:full-stage} implies that for every time step $t$, $X^t$ is $(1-a)$-envy-bounded and has no envy cycles, and that $x^t$ is a maximum weight matching. Thus by Lemma~\ref{lem:envy-bounded}, $X^t$ is EF1 for all times $t$. Thus Algorithm~\ref{alg:sym-bin} is EF1-over-rounds.

For the time bound, Lemma~\ref{lem:termination} tells us that the while loop has $O(n^2)$ iterations, and each iteration is constant time. Since our graph is bipartite, finding the initial maximum weight matching can be done in $O(n^{2.5})$ time~\cite{Micali1980}. Therefore Algorithm~\ref{alg:sym-bin} runs in time $O(n^{2.5} + n^2) = O(n^{2.5})$ for each time step.
\end{proof}

\section{Counterexamples}\label{sec:counter}

A natural question is whether Theorem~\ref{thm:sym-bin} can be extended to all dynamic binary valuations (i.e., not necessarily symmetric). The answer is unfortunately no, which we show in two different ways. First, for dynamic binary valuations, EF1-over-rounds alone is impossible (Theorem~\ref{thm:dyn-bin-counter}). Second, for non-dynamic binary valuations, it is impossible to guarantee both EF1-over-rounds and maximum weight matching for each time step (Theorem~\ref{thm:bin-max-counter}).

\begin{figure}[ht!bp]
\centering
\begin{tikzpicture}[->,thick,every node/.style={draw,circle,minimum size=0 mm, node distance=1.5 cm, font=\footnotesize, fill=blue!20}]
  \node (1) at (0,0) {$1$};
  \node (3) [right of=1] {$3$};
  \node (2) [below of=1] {$2$};
  \node (4) [right of=2] {$4$};
  \node (5) at (5,0) {$1$};
  \node (7) [right of=5] {$3$}; 
  \node (6) [below of=5] {$2$}; 
  \node (8) [right of=6] {$4$};
  \node[draw=none,fill=none] (t1) at (.7, 1) {$t = 1$};
  \node[draw=none,fill=none] (t2) at (5.7, 1) {$t = 2$};
  
  \foreach \i [remember=\i as \lasti (initially 1)] in {2,3,4,1} {\draw (\lasti) -- (\i);};\
\draw[<->] (5) -- (6);
\draw (7) -- (6);
\draw (8) -- (5);
  
\end{tikzpicture}
\caption{An instance with dynamic and binary valuations where EF1-over-rounds is impossible.}\label{fig:dyn-bin-counter}
\end{figure}
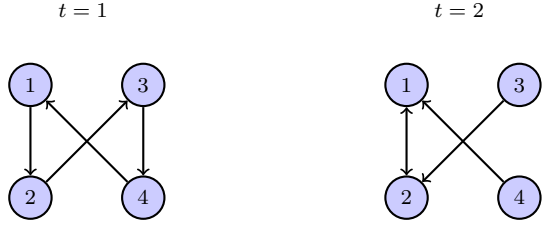

 \begin{theorem}\label{thm:dyn-bin-counter}
For dynamic and binary valuations, there is no algorithm which is EF1-over-rounds.
 \end{theorem}
 
 \begin{proof}
 Consider the instance given by Figure~\ref{fig:dyn-bin-counter}. We have $N = \{1,3\}$ and $M = \{2,4\}$, and an edge from agent $i$ to agent $j$ indicates $v_i(j) = 1$. There are only two perfect matchings for us to choose from: $\{1,2\}$ and $\{3,4\}$, or $\{1,4\}$ and $\{3,2\}$. For $t=1$, by symmetry, suppose we match $\{1,2\}$ and $\{3,4\}$. Then under $X^1$ (i.e., after these matches), agent $2$ envies agent $4$ and vice versa. For $t=2$, agents 2 and 4 both want to be matched with agent 1, and exactly one of them will be. Thus whichever of them is not matched to agent 1 on $t=2$ will envy the other in violation of EF1. This assumes a sufficient choice of $a$ for $v_i(j) \in \{a,1 \}$; $a=0$ suffices.
 \end{proof}

 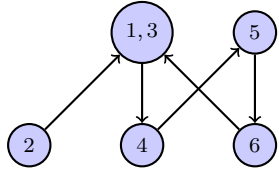
\begin{figure}[ht!bp]
\centering
\begin{tikzpicture}[->,thick,every node/.style={draw,circle,minimum size=0 mm, node distance=1.5 cm, font=\footnotesize,fill=blue!20}]
  \node[draw=none,fill=none] (1) at (0,0) {$$};
  \node (3) [right of=1] {$1,3$};
  \node (2) [below of=1] {$2$};
  \node (4) [right of=2] {$4$};
  \node (5) [right of=3] {$5$};
  \node (6) [below of=5] {$6$}; 
  
  \foreach \i [remember=\i as \lasti (initially 3)] in {4,5,6,3} {\draw (\lasti) -- (\i);};\
  \draw (2) -- (3);
  
\end{tikzpicture}
\caption{An instance with binary valuations where guaranteeing both EF1-over-rounds and maximum weight matching is impossible.}\label{fig:match-bin-counter}
\end{figure}

For some intuition behind Figure~\ref{fig:match-bin-counter}, note that are two cycles of desire: $(1,4,5,6)$ and $(3,4,5,6)$. Like in the previous counterexample, these cycles will cause problems, but here we have the additional consideration that agents 1 and 3 are competing for agent 4. We show that the frequency with which agents 4 and 5 are matched is at least the frequency with which \emph{either} agent 1 or agent 3 is matched with agent 4. For example, if agents 1 and 3 have each been matched to agent 4 twice, then agents 4 and 5 will have been matched 4 times. This leads to agents 1 and 3 increasingly envying agent 5, until EF1 is violated. 

The assumption of maximum weight is necessary only to prevent agents 2 and 5 from ever being matched: if agents 2 and 5 can be matched, the above argument can be circumvented.

\begin{theorem}\label{thm:bin-max-counter}
For binary valuations, there is no algorithm which is EF1-over-rounds and also chooses a maximum weight matching for each time step, even for non-dynamic valuations.
\end{theorem}

\begin{proof}
Consider the instance given by Figure~\ref{fig:match-bin-counter}. Here $N = \{1,3,5\}$ and $M = \{2,4,6\}$, and an edge from agent $i$ to agent $j$ indicates $v_i(j) = 1$. For this example, we will assume $v_i(j) \in \{0,1\}$ for all $i,j$. We represent the pair $\{1,3\}$ as a single node in Figure~\ref{fig:match-bin-counter} to indicate that those are identical agents, i.e., $v_1(i) = v_3(i)$ and $v_i(1) = v_i(3)$ for all $i \in M$.

On each time step, we must choose a perfect matching between $N$ and $M$. Recall that the weight of a matching $x$ is $\frac{1}{2} \sum_{\{i,j\} \in x} (v_i(j) + v_j(i))$. We claim that any maximum weight matching in this instance has weight 1.5. A weight of above 1.5 is impossible, since that would require matching two agents who both like each other; no such pair of agents exists in this instance. The matching $\{\{1,2\}, \{3,4\}, \{5,6\}\}$ has weight 1.5, so a weight of 1.5 is possible. Thus any maximum weight time step matching has weight 1.5.

We are going to be interested in ever \emph{pair} of matches each agent receives, so we will look at $t = 2r$ for integers $r$. First, we claim that whenever $t=2r$ for an integer $r$, the following hold:
\begin{enumerate}
\item $v_4(X_4^t) = v_4(X_6^t)$
\item $(v_1(X_5^t) - v_1(X_1^t)) + (v_3(X_5^t) - v_3(X_3^t)) = r$
\end{enumerate}

We proceed by induction on $r$. The claims trivially hold when $r = 0$, since $X_i^0 = \emptyset$ for all $i$. Thus assume the claims hold for $t=2r$. At $t' = 2(r+1) = 2r + 2$, each agent will have received two new matches; let $Y$ be the multiset of new pairings during this period.

First note that agents 2 and 5 can never be matched, since that would result in a matching weight of 2 for some time step. Thus agent 5 can only be matched to agents 4 and 6. We claim that $\{5,4\}$ and $\{5,6\}$ each appear in $Y$ exactly once. If $\{5,6\}$ appeared twice, we would have $v_4(X_6^{t'}) - v_4(X_4^{t'}) = 2$, which would violate EF1. If $\{5,4\}$ appeared twice, $v_1(X_1^t)$ and $v_3(X_3^t)$, while $v_1(X_5^t)$ and $v_3(X_5^t)$ both increased by 2. Thus we would have $v_i(X_5^{t'}) - v_i(X_i^{t'}) \ge 2$ for at least one $i \in \{1,3\}$, which would violate EF1.

Thus $\{5,4\}$ and $\{5,6\}$ each appear in $Y$ exactly once. This immediately satisfies $v_4(X_4^{t'}) = v_4(X_6^{t'})$, since agent 4 only likes agent 5. Also, agent 4 will be matched to exactly one of agents 1 and 3. Thus $v_1(X_1^t) + v_3(X_3^t)$ increases by 1, while $v_1(X_5^t)$ and $v_3(X_5^t)$ each increase by 1. Thus $(v_1(X_5^t) - v_1(X_1^t)) + (v_3(X_5^t) - v_3(X_3^t))$ increases by 1, as required.

This completes the induction. Claim 2 implies that at $t=6$, we have $(v_1(X_5^t) - v_1(X_1^t)) + (v_3(X_5^t) - v_3(X_3^t)) = 3$. But then we have $v_i(X_5^t) - v_i(X_i^t) = 2$ for some $i \in \{1,3\}$, which implies that $X^t$ is not EF1. But we know that $|X_i^t| = |X_j^t|$ for all $i,j$, so this contradicts the claim that $\X$ is EF1-over-rounds.
\end{proof}

\section{Beyond symmetric valuations}\label{sec:asym-bin}

%


In this section, we show that it is possible to go beyond symmetric binary valuations, at least a bit: we give an algorithm that is EF1-over-time for binary $\{0,1\}$ valuations under the assumption that every ``cycle" in those valuations is symmetric. We define ``cycle" and ``symmetric cycle" later, but this assumption generalizes symmetric binary $\{0,1\}$ valuations. It is important to note that the result of this section holds only for binary $\{0,1\}$ valuations, not all $\{a,1\}$ valuations for any $a \in [0,1]$. However, recall that both of the counterexamples given in Section~\ref{sec:counter} also use $\{0,1\}$ valuations.

Since we are considering EF1-over-time, $x^t$ will always contain a single match, i.e., $x^t = \{\{i,j\}\}$. We will abuse notation and simply write $x^t = \{i,j\}$, as in the pseudocode from Algorithm~\ref{alg:asym-bin}.

\begin{algorithm*}[htb]
\centering
\begin{algorithmic}[1]
\Function{EF1Matching}{$N, M, \V$}
\ForAll{$t\in \mathbb{N}_{> 0}$}
    \State $x^t \gets$ any pairing $(i,k)$ s.t. $v_i(k) + v_k(i) \ge 1$
    \While{$\exists$ agents $i,j$ s.t. $v_i(X_j^t) - v_i(X_i^t) > 1$}
    	\State $x^t \gets x^t\cup\{i\}\setminus\{j\}$
    \EndWhile
    \State \texttt{MakeMatch}($x^t$)
  \EndFor
\EndFunction
\end{algorithmic}
\caption{An EF1-over-time algorithm for binary valuations with only symmetric cycles.}
\label{alg:asym-bin}
\end{algorithm*}

Our algorithm is very simple. On each time step, we propose a pairing $\{i,k\}$ that at least one of $i$ and $k$ is happy with. Then while there exists an agent $j$ such that making this match would violate EF1 from $j$'s perspective, we swap $j$ into the pairing for this time step. We will need to show that the while loop of Algorithm~\ref{alg:asym-bin} always terminates. As in Section~\ref{sec:sym-bin}, the crucial step will be showing that the cumulative matching can never have an envy cycle. We will show that the presence of an envy cycle would contradict the assumption of ``only symmetric cycles".

To define ``symmetric cycles", we must first define the \emph{desire graph}.

\begin{definition}
For a (non-dynamic) valuation profile $\V$, the \emph{desire graph} is an undirected graph with a vertex for each agent in $N\cup M$, and an edge from $i$ to $k$ if $v_i(k) + v_k(i) \ge 1$.
\end{definition}

We say that an edge $\{i,k\}$ in the desire graph is symmetric if $v_i(k) = v_k(i) = 1$, and a cycle in the desire graph is symmetric if every edge is symmetric.

\begin{lemma}\label{lem:1-bounded}
Assume that $v_i(k) \in \{0,1\}$ for any pair of agents $i,k$, and that at time $t$, $v_i(X_j^t) - v_i(X_i^t) \le 1$ for all agents $i,j$. Then $X^t$ is EF1.
\end{lemma}

\begin{proof}
Suppose $i$ envies $j$ under $X^t$. Then there must exist $k \in X_j^t$ such that $v_i(k) = 1$. Therefore
\[
v_i(X_j^t\setminus\{k\}) = v_i(X_j^t) - 1 \le v_i(X_i^t)
\]
Thus $X^t$ is EF1.
\end{proof}

\begin{lemma}\label{lem:asym-bin-terminates}
Suppose that $X^{t-1}$ has no envy cycles. Then the while loop for time step $t$ has at most $n + m$ iterations before terminating.
\end{lemma}

\begin{proof}
Consider an arbitrary iteration of the while loop, and let $i,j$ be as defined in Algorithm~\ref{alg:asym-bin}. First, we claim that $i$ envies $j$ under $X^{t-1}$. Suppose the opposite: since $i$'s value for $j$'s bundle can increase by at most 1 on each time step, we have $v_i(X_j^t) - v_i(X_i^t) \le 1$, which contradicts the while loop condition. Therefore $i$ envies $j$ under $X^{t-1}$.

Thus on every iteration of the while loop, a match $\{j,k\}$ changes to $\{i,k\}$ for an agent $i$ who envies $j$. Note that the match always involves one agent in $N$ and one agent in $M$. Since there are no envy cycles in $X^{t-1}$ by assumption, the agent in $N$ in $x^t$ can change at most $n$ times, and the agent in $M$ in $x^t$ can change at most $m$ times. Thus the while loop terminates after at most $n+m$ iterations.
\end{proof}

To prove correctness, we will use the notion of a \emph{circuit} in the envy graph.

\begin{definition}
Given a graph, a \emph{circuit} is a sequence of (not necessarily distinct) vertices $(u_1, u_2, \dots, u_\ell)$ such that for each $i \in \{1\dots \ell\}$, the edge $(u_i, u_{i\bmod \ell + 1})$ is in the graph.
\end{definition}

In words, a circuit is a cycle that is allowed to repeat vertices. For a given circuit of length $\ell$, define $\suc(i) = i\bmod \ell + 1$. For a graph without self-loops (such as the envy graph), for each $i$, $u_i$ and $u_{\suc(i)}$ are distinct vertices.

Note that in the following lemma statement, $i$ itself is a vertex, not $u_i$. This means $i$ is independent of the labeling of $C$, which is useful because Algorithm~\ref{alg:circuit} changes said labeling. For Lemma~\ref{lem:circuit}, we will use $\suc(i) = i\bmod \ell + 1$.

\begin{lemma}\label{lem:circuit}
Let $C = (u_1, u_2, \dots, u_\ell)$ be a circuit in a graph $G$ without self-loops. Then for each $i \in C$, there exists a cycle in $G$ which contains the edge $(i, \suc(i))$.
\end{lemma}

\begin{proof}
We claim that Algorithm~\ref{alg:circuit} returns a cycle containing the edge $(i, \suc(i))$. We will use $C$ to refer to the original circuit, and $D$ to refer to the variable in the algorithm. First, note that since at least one vertex is removed from $D$ on each iteration of the while loop, the algorithm is guaranteed to terminate.

\textbf{Claim 1: $i = u_{|D|}$ and $\suc(i) = u_1$ after every iteration.} Note that this claim implies that the edge $(i, \suc(i))$ always remains in $D$. To prove the claim, we trivially always have $\suc(i) = u_1$, since we never remove or relabel $u_1$. For $i = u_{|D|}$, if $k < |D|$, the claim trivially holds, since $u_{|D|}$ is unaffected. If $k = |D|$, then the last vertex in the new $D$ is $u_j$, which is equal to $u_k$ (and thus equal to $u_{|D|}$) by construction. Thus $i = u_{|D|}$ and $\suc(i) = u_1$ after every iteration, and the edge $(i, \suc(i))$ always remains in $D$.

\textbf{Claim 2: The algorithm returns a circuit $D$.} Initially $D$ is a circuit by assumption. Suppose that before some iteration of the while loop, $D=(u_1, u_2,\dots, u_{|D|})$ is a circuit. Let $j,k$ be defined as in Algorithm~\ref{alg:circuit}; we claim that $(u_1, u_2,\dots, u_j, u_{k+1},\dots, u_{|D|})$ is a circuit. 

To do this, we must show that each of these edges is in the graph. The edge $(u_{|D|}, u_1)$ must be in the graph by Claim 1. For each $q \in \{1\dots j-1, k+1\dots |D| -1$, the edge $(u_q, u_{q+1})$ is also in $D=(u_1, u_2,\dots, u_{|D|})$. Since $D$ is a circuit, this edge must be in the graph. Finally, we consider the edge $(u_j, u_{k+1})$. Since $D$ is a circuit, we know that $(u_k, u_{k+1})$ must be in the graph. Since $u_j = u_k$, $(u_j, u_{k+1})$ is in the graph. Therefore $(u_1, u_2,\dots, u_j, u_{k+1},\dots, u_{|D|})$ is a circuit.

Thus the algorithm is guaranteed to return a circuit. Furthermore, the circuit returned by the algorithm must consist of distinct vertices: if any vertex repeated, the while loop would not have terminated. Thus the algorithm returns a cycle. Finally, by Claim 1, this cycle contains the edge $(i, \suc(i))$, as required.
\end{proof}

\begin{algorithm*}[htb]
\centering
\begin{algorithmic}[1]
\Function{FindCycleInCircuit}{$C = (u_1\dots u_\ell), i$}
\State relabel $D = (u_1, u_2, \dots, u_\ell)$ so that $i = u_\ell$ and $\suc(i) = u_1$
\While{$\exists j,k \in \{1\dots |D|\}$ such that $u_j = u_k$} \Comment{Assume WLOG that $j < k$}
	\State $D \gets (u_1, u_2,\dots, u_j, u_{k+1},\dots, u_{|D|})$ \Comment{Remove $u_{j+1}\dots u_k$}
	\State relabel $D = (u_1, u_2,\dots u_{|D|})$
\EndWhile
\Return $D$
\EndFunction
\end{algorithmic}
\caption{Given a circuit $C$ containing a vertex $i$, this algorithm finds a cycle containing the edge $(i, \suc(i))$.}
\label{alg:circuit}
\end{algorithm*}

\begin{theorem}\label{thm:asym-bin}
Assume that $v_i(k) \in \{0,1\}$ for any pair of agents $i,k$, and that every cycle in the desire graph is symmetric. Then Algorithm~\ref{alg:asym-bin} is EF1-over-time.
\end{theorem}

\begin{proof}
We claim by induction on $t$ that following hold for each $X^t$: (1) $v_i(X_j^t) - v_i(X_i^t) \le 1$ for all agents $i,j$, and (2) there are no envy cycles in $X^t$. Importantly, Claim 2 in combination with Lemma~\ref{lem:asym-bin-terminates} will imply that each while loop does terminate, so the sequence $\X$ is well-defined.

At $t = 0$, $v_i(X_j^t) = 0$ for all $i,j$, so both claims follow trivially (and there is no while loop for $t=0$). Thus assume that both claims hold at an arbitrary time $t-1$. Since we do not exit the while loop until Claim 1 is satisfied, Claim 1 trivially holds at time $t$. Suppose that Claim 2 does not hold at time $t$, i.e., there exists an envy cycle in $X^t$. 

Let $C = (1, 2\dots {|C|})$ be an envy cycle in $X^t$, and let $\suc(i) = i \bmod |C| + 1$. Since $C$ is an envy cycle, for each $i \in C$, $i$ envies agent $\suc(i)$. Note that all agents in $C$ must be on the same side of the market; let $S$ be that side. We claim for each $i \in C$, there exists an agent $j_i\not\in S$ such that the edges $(i, j_i)$ and $(\suc(i), j_i)$ are both in the desire graph. Suppose this is not the case for some $i$: then for every agent $j \not \in S$ that agent $i$ likes, neither agent $j$ nor agent $\suc(i)$ are interested in each other. Therefore agent $\suc(i)$ would never be matched with any agent whom $i$ likes, so agent $i$ cannot envy agent $\suc(i)$. 

Let $J_i = \{j \not\in S: \text{$(i, j)$ and $(\suc(i), j)$ are in the desire graph}\}$. Then for any $j_1 \in J_1, j_2\in J_2\dots j_{|C|} \in J_{|C|}$, $(1, j_1, 2, j_2\dots {|C|/2}, j_{|C|/2})$ is a circuit\footnote{The reason vertices may be repeated is that $J_i \cap J_{i'}$ may not be empty, even for $i\ne i'$.} in the desire graph. Thus by Lemma~\ref{lem:circuit}, for all $i$ and all $j\in J_i$, $(i, j)$ and $(\suc(i), j)$ each belong to a cycle in the desire graph. Therefore all such edges must be symmetric, by assumption.

We claim that as a consequence, for all $i$, $v_\suc(i)(X_\suc(i)^t) > v_{i}(X_{i}^t)$. Without loss of generality, fix $i=1$ for brevity. We know that we only ever make a match $(i,j)$ if that edge is in the desire graph. Consider an agent $j \in X_2^t$ such that $v_1(j) = 1$: then the edges $(1,j)$ and $(2,j)$ are both in the desire graph. Since these edges must then by symmetric, we have $v_2(j) = 1$.

Therefore for every $j \in X_2^t$ such that $v_{1}(j) = 1$, we have $v_{2}(j) = 1$. This implies that $v_{2}(X_{2}^t) \ge v_{1}(X_{2}^t)$. Since $v_{1}(X_{2}^t) > v_{1}(X_{1}^t)$ (agent $1$ must envy agent $2$, because $C$ is an envy cycle), we have $v_{2}(X_{2}^t) > v_{1}(X_{1}^t)$. By symmetry, this holds for all $i$. Applying this relationship around the cycle, we get $v_{1}(X_{1}^t) > v_{1}(X_{1}^t)$, which is a contradiction. We conclude that there is no envy cycle in $X^t$.

This completes the induction. Thus at every time $t$, for every pair $i,j$, $v_i(X_j^t) - v_i(X_i^t) \le 1$. Then by Lemma~\ref{lem:1-bounded}, $X^t$ is EF1 for all $t$. Therefore $\X$ is EF1-over-time.
\end{proof}

\section{Beyond binary valuations}\label{sec:n=2}

In this section, we depart from our assumption that $|N| = |M|$, and consider the case when one side of the market (without loss of generality, say $N$) has two agents. Alternatively, one can think of $N$ as having two non-dummy agents and $n-2$ dummy agents, with $M$ having $n$ agents.

In this case, for additive (not necessary binary) valuations, we give a matching algorithm which produces a sequence $\X$ which is EF1-over-time. Our algorithm is based on the round-robin algorithm for EF1 in one-sided markets with additive valuations, where agents take turns choosing their favorite remaining item from an available pool. As in Section~\ref{sec:asym-bin}, we abuse notation and write $x^t = \{i,j\}$ instead of $x^t = \{\{i,j\}\}$.

Algorithm~\ref{alg:n=2} divides time into stages. Each stage will match each $i \in M$ exactly twice: in particular, once to agent $1 \in N$ and once to agent $2 \in N$. This also means that each stage matches each agent in $N$ to every agent in $M$ exactly once. The key consequence is that at the end of each stage, the cumulative matching is fully envy-free. The only remaining technical consideration is what happens within a stage.

Each stage has two phases. In the first phase, agents 1 and 2 in $N$ alternate picking their favorite agent in $M$ without replacement, starting with agent 1. We record the order of these matches in $\sigma$. In the second phase, we match agents in $M$ in the same order as in phase 1, but this time, we start with agent 2. This ensures that agent 1 gets matched to exactly those agents in $M$ she was not matched to during phase 1 (and similarly for agent 2). For a given sequence of matches $\sigma$, let $X(\sigma)$ be the corresponding matching. For example, if $\sigma = ((1,1), (2,2), (1,3))$, then $X_1(\sigma) = \{1,3\}$ and $X_2(\sigma) = \{2\}$.

\begin{algorithm*}[htb]
\centering
\begin{algorithmic}[1]
\Function{EF1Matching}{$N, M, \V$}
\State $t \gets 1$
\ForAll{$s\in \mathbb{N}_{\ge 0}$}
	\State $P \gets M$ \Comment{Initially, the pool contains all agents in $M$.}
	\State $\sigma \gets []$\
	\State $i\gets 1$
	\While{$P\ne \emptyset$} \Comment{Until everyone in $M$ has been matched once,}
		\State $j \gets \argmax_{k \in P} v_i(k)$ \Comment{agent $i$ chooses her favorite unmatched agent in $M$.}
		\State $x^t \gets \{i, j\}$ \Comment{Match those agents,}
		\State $\texttt{MakeMatch}(x^t)$
		\State $P \gets P\setminus\{j\}$ \Comment{remove the matched agent in $M$ from the pool,}
		\State $\sigma.\texttt{append}((i, j))$ \Comment{record this match,}
		\State $i \gets (i+1)\bmod 2$ \Comment{and switch to the other agent in $N$.}
		\State $t++$
	\EndWhile
	\State $i\gets 2$
	\ForAll{$z \in \sigma$} \Comment{Go through the matches we already made in order,}
		\State $(\_, j) \gets z$ \Comment{take the agent in $M$ from each match,}
		\State $x^t \gets \{i, j\}$ \Comment{and match her to the agent in $N$ she was not matched to before.}
		\State $\texttt{MakeMatch}(x^t)$
		\State $i \gets (i+1) \bmod 2$
		\State $t++$
	\EndFor
  \EndFor
\EndFunction
\end{algorithmic}
\caption{An algorithm that is EF1-over-time when one side of the market has two agents.}
\label{alg:n=2}
\end{algorithm*}

Lemma~\ref{lem:sigma} is the key to our algorithm. For those familiar with the round robin algorithm in one-sided fair division, the idea is the same.

\begin{lemma}\label{lem:sigma}
If $\sigma'$ is a prefix or suffix of $\sigma$, then $X(\sigma')$ is EF1 with respect to $N$.
\end{lemma}

\begin{proof}
We know that for each match $(i,j)$ that occurs in $\sigma$ (say $i \in N$ and $j \in M$), $i$ prefers $j$ to all agents in $M$ occurring after $j$ in $\sigma$ (since $j$ was $i$'s favorite remaining agent). Now consider an arbitrary $i \in N$, and let $\sigma'$ be any prefix or suffix of $\sigma$. If the first match in $\sigma'$ does not involve agent $i$, let $\sigma''$ be the suffix of $\sigma'$ obtained by removing the first match in $\sigma'$; otherwise, let $\sigma'' = \sigma'$. Thus the odd-numbered matches in $\sigma''$ belong to agent $i$, and the even-numbered matches in $\sigma''$ belong to the other agent in $N$.

Let $\sigma''_k$ denote the agent in $M$ in the $k$th match in $\sigma''$. Then we have $v_i(\sigma''_1) \ge v_i(\sigma''_2)$, $v_i(\sigma''_3) \ge v_i(\sigma''_4)$, and so on. Thus by additivity of valuations, agent $i$ does not envy the other agent under $X(\sigma'')$. We conclude that $X(\sigma')$ is EF1 with respect to $N$. 
\end{proof}

\begin{theorem}\label{thm:n=2}
Algorithm~\ref{alg:n=2} is EF1-over-time.
\end{theorem}

\begin{proof}
Let $\X = X^1, X^2\dots$ be the sequence of cumulative matchings. Let $t_s$ be the time at which stage $s$ begins: $t_1 = 0$, and $t_{s+1} = t_s + 2m$. We first show by induction that $X^{t_s}$ is fully envy-free for each stage $s$. Since $X^{t_1}_i = X^0_i = \emptyset$ for all $i$, this is trivially true for $s=1$. Now assume the claim holds for some stage $s$. During stage $s$, each agent in $N$ is matched to each agent in $M$ exactly once, and each agent in $M$ is matched to each agent in $N$ exactly once. Thus on stage $s$, each agent receives exactly the same set of matches. Since valuations are additive, we have $v_i(X_i^{t_{s+1}}) - v_i(X_i^{t_s}) = v_i(X_j^{t_{s+1}}) - v_i(X_j^{t_s})$ for any $i,j \in N$ or $i,j \in M$. Since $X^{t_s}$ is envy-free by assumption, we have $v_i(X_i^{t_s}) \ge v_i(X_j^{t_s})$; combined with the previous equation, this implies that $X^{t_{s+1}}$ is envy-free. This completes the induction.

\textbf{Claim 1: $X^t$ is EF1 with respect to $M$ at each time $t$.} On each stage, each agent in $M$ is matched to each agent in $N=\{1,2\}$ exactly once. Furthermore, every agent receives one match within the stage before any agent receives two. Therefore for $i,j \in M$, $i$'s envy for $j$ can always be eliminated by removing the one match $j$ has that $i$ does not have yet in stage $j$; if such a match does not exist, $i$ cannot envy $j$.

\textbf{Claim 2: $X^t$ is EF1 with respect to $N$ at all times $t$.} We know that at the beginning of each stage, $X^t$ is fully envy-free. Let $s$ be the stage during which $t$ occurs. During phase 1 of stage $s$, each agent $i$'s cumulative matching is $X_i^t = X_i^{t_s} \cup X_i(\sigma')$, for a prefix $\sigma'$ of $\sigma$. By Lemma~\ref{lem:sigma}, $X(\sigma')$ is EF1 with respect to $N$, so by additivity of valuations, $X^t$ is EF1 with respect to $N$. For each $i \in N = \{1,2\}$, let $\bar{i}$ denote the other agent in $N$. During phase 2, $i$'s cumulative matching is $X_i^t = X_i^{t_s} \cup X_i(\sigma) \cup X_{\bar{i}}(\sigma')$ for a prefix $\sigma'$ of $\sigma$. The difference between $X_i(\sigma) \cup X_{\bar{i}}(\sigma')$ and $X_{\bar{i}}(\sigma) \cup X_i(\sigma')$ is $X_i(\sigma'')$ for a \emph{suffix} $\sigma''$ of $\sigma$. Again by Lemma~\ref{lem:sigma}, $X(\sigma'')$ is EF1 with respect to $N$, so again by additivity of valuations, $X^t$ is EF1 with respect to $N$. Thus at all time steps $t$ during each stage $s$, $X^t$ is EF1 with respect to $N$.

Claims 1 and 2 together imply that $X^t$ is EF1 for all times $t$. Thus $\X$ is EF1-over-time, as required.
\end{proof}

\section{Conclusion}\label{sec:conclusion}

In this paper, we proposed a model of envy-freeness for repeated two-sided matching. For binary and symmetric valuations, we gave an algorithm that (1) satisfies EF1-over-rounds, (2) chooses a maximum weight matching for each time step, and (3) works even for dynamic valuations (Section~\ref{sec:sym-bin}). We also showed that without symmetry, (1) + (2) together and (1) + (3) together are each impossible. We will also gave algorithms for the cases where valuations are binary with ``only symmetric cycles" (Section~\ref{sec:asym-bin}), and when one side of the market has only two agents (Section~\ref{sec:n=2}). We are not aware of any prior models of almost envy-freeness for two-sided markets.

Our negative results for even binary valuations suggest that EF1-over-rounds may be too much to ask for. However, our results do not rule out the possibility of EF1-over-time, even for general additive valuations. More broadly, future work could investigate other possible fairness notions for this setting.

Another possible future direction concerns more general study of two-sided preferences. Envy-freeness is an example of a topic that has been widely studied for one-sided resource allocation, but not for two-sided markets. We wonder if there are other such topics that are worthy of study for two-sided preferences.



\bibliographystyle{plain}
\bibliography{refs}

\end{document}